\documentclass[10pt,twocolumn]{article}
\usepackage[margin=0.75in]{geometry}
\usepackage{amsmath,amssymb,amsthm}
\usepackage[hidelinks]{hyperref}
\usepackage{cite}
\usepackage{esvect}
\usepackage{enumitem}
\usepackage{booktabs}
\usepackage{algorithm}
\usepackage[noend]{algpseudocode}
\usepackage{tikz}
\usepackage{pgfplots}
\pgfplotsset{compat=1.18}
\usetikzlibrary{patterns,positioning,pgfplots.fillbetween}

\newtheorem{theorem}{Theorem}
\newtheorem{lemma}[theorem]{Lemma}
\newtheorem{corollary}[theorem]{Corollary}
\newtheorem{definition}[theorem]{Definition}
\newtheorem{remark}[theorem]{Remark}

\newcommand{\Trap}{\textsc{TRAP}}
\newcommand{\Snare}{\textsc{SNARE}}
\newcommand{\BFTCR}{\textsc{BFTCR}}
\newcommand{\GG}{\mathcal{G}}
\newcommand{\RR}{\mathcal{R}}
\newcommand{\LL}{\mathcal{L}}
\newcommand{\SSS}{\mathcal{S}}
\usepackage[normalem]{ulem}

\title{\textbf{WIP: \Snare{}: A TRAP for Rational Players to Solve\\Byzantine Consensus in the $\boldsymbol{5f{+}1}$ Model}}
\author{Alejandro Ranchal-Pedrosa\\ \textit{Sei Labs}\\ \texttt{alex@seinetwork.io}
\and
Benjamin Marsh\\ \textit{Sei Labs and University of Portsmouth}}
\date{}

\begin{document}
\maketitle

\begin{abstract}
The \Trap{} protocol solves rational agreement by combining accountable consensus (which \emph{predecides}) with a one-shot \BFTCR{} finalization phase (which \emph{decides}). We present \Snare{} (\emph{Scalable Nash Agreement via Reward and Exclusion}), the adaptation of \Trap{} to $n=5f{+}1$, and prove $\epsilon$-$(k,t)$-robustness for rational agreement tolerating coalitions up to ${\approx}73\%$ with deposits under $0.5\%$ of the gain.

A central finding is that appending a single all-to-all broadcast round with the $4f{+}1$ threshold after predecisions yields $\epsilon$-$(k,t)$-robustness for coalitions up to $3f$ (${\approx}60\%$) \emph{without any deposit}: we need not model or know the utility function of deviating players, only that they participate in the protocol. These players can be \emph{deceitful} (arbitrary unknown utility), not just rational, and the finalization structure prevents disagreement regardless of their motivation. This observation is protocol-agnostic, applies to any $5f{+}1$ protocol at the cost of one message delay that runs concurrently with the next view, and does not require commit-reveal mechanisms. Above $60\%$, the full baiting mechanism with deposits under $0.5\%$ extends tolerance to ${\approx}73\%$.

A second finding is that valid-candidacy, the property preventing reward front-running, holds unconditionally regardless of the quorum threshold, removing both the $n>2(k{+}t)$ and $n>\frac{3}{2}k{+}3t$ constraints from the original \Trap{}. This retroactively extends the $3f{+}1$ bound from $C<n/2$ to $C<5n/9$. The binding constraint in both models is the winner consensus operating on $2f$ residual players after excluding $3f{+}1$ detected equivocators. We explore avenues for relaxing this limit.
\end{abstract}

\section{Introduction}\label{sec:intro}

The \Trap{} protocol of Ranchal-Pedrosa and Gramoli~\cite{TRAP22} solves the rational agreement problem under partial synchrony by introducing a baiting strategy that rewards coalition members for betraying their coalition before a disagreement is finalized. The protocol separates an accountable consensus layer producing \emph{predecisions} from a one-shot Byzantine Fault Tolerant Commit-Reveal (\BFTCR{}) phase producing \emph{decisions}. In the $n=3f{+}1$ model, $\epsilon$-$(k,t)$-robustness holds for $n>\max(\frac{3}{2}k{+}3t,\;2(k{+}t))$, tolerating coalitions up to $n/2$ and at most double-spending.

We present \Snare{}, the adaptation of \Trap{} to $n=5f{+}1$ ($t_0=f$, quorum $h=4f{+}1$). Two structural features reshape the landscape.

First, the same-view partition bound~\cite{RG21} yields $C_{\textsf{fin}}(2)=3f{+}1\approx 60\%$, creating a large \emph{no-fork zone} ($C\leq 3f$) where the finalization phase cannot produce conflicting decisions. This zone tolerates \emph{deceitful} players~\cite{Basilic23} with arbitrary unknown utility functions, not just rational players whose utility we model. The no-fork property follows from a single all-to-all broadcast round with the $4f{+}1$ threshold after predecisions. This observation is protocol-agnostic and applies to \emph{any} $5f{+}1$ protocol at the cost of one message delay off the critical path. Notably, the $60\%$ safety threshold is nearly double the $33\%$ of two-round voting in the standard $3f{+}1$ model. The finalization round also provides natural resistance to long-range attacks: cross-view disagreements (where an adversary equivocates across view changes to create conflicting predecisions) are subsumed by the single-shot finalization, which cannot fork below $60\%$.

Second, valid-candidacy holds unconditionally in the $5f{+}1$ model. The constraints $n>2(k{+}t)$ and $n>\frac{3}{2}k{+}3t$ in the original \Trap{} (Lemma~4.3 of~\cite{TRAP22}) ensure that baiters' commitments reach enough honest players for valid proofs-of-baiting. We show (Lemma~\ref{lem:vc}) that both are unnecessary: the fork mechanism itself delivers the baiters' encrypted commitments to all honest players, reducing the constraint to $n{-}C\geq f{+}1$, which is subsumed by the winner consensus feasibility condition. This has immediate implications for the original $3f{+}1$ \Trap{}: removing both constraints extends the achievable coalition from $C<n/2$ to $C<5n/9\approx 56\%$.

With valid-candidacy removed, the binding limit becomes the \emph{winner consensus} among $n'=2f$ residual players after excluding $3f{+}1$ detected equivocators, requiring the remaining coalition below $n'/3$. We explore avenues for relaxing this bottleneck. The paper is organized around the branch count $a$: no-fork ($a=1$, no deposits), double-spend ($a=2$), triple-spend ($a=3$), and the winner consensus ceiling.

\subsection{Contributions}

(1) We observe that one extra all-to-all broadcast round, outside the critical path of view-change consensus and executed only once per decision, with threshold $4f{+}1$ increases safety $3\times$ in any $5f{+}1$ protocol, from ${\approx}20\%$ to ${\approx}60\%$, at zero financial cost. In blockchains and other repeated consensus settings, this round can be piggybacked onto the next consensus iteration: clients or validators can simply wait one additional message delay after deciding before performing off-chain actions. (2) We derive the no-fork regime $C\leq 3f$ (${\approx}60\%$) where no deposit is needed and which tolerates even deceitful players. (3) We prove valid-candidacy (Lemma~\ref{lem:vc}), removing the constraints $n>2(k{+}t)$ and $n>\frac{3}{2}k{+}3t$ from~\cite{TRAP22} and extending the coalition to ${\approx}73\%$. This also applies retroactively to $3f{+}1$. (4) We derive the full analysis for double-spending (Section~\ref{sec:double}) and triple-spending (Section~\ref{sec:triple}) with deposits under $0.5\%$ of the gain. (5) We identify the winner consensus as the binding constraint and discuss paths to relax it (Section~\ref{sec:discussion}). (6) As a byproduct, we show that the original $3f{+}1$ \Trap{} can be retroactively extended from $C<n/2$ to $C<5n/9$, tolerating triple-spending, by removing both valid-candidacy constraints $n>2(k{+}t)$ and $n>\frac{3}{2}k{+}3t$ (Section~\ref{sec:vc}).

\section{Related Work}\label{sec:related}

\paragraph{Byzantine consensus.}
Consensus under partial synchrony requires $n>3t$~\cite{DLS88}. The $5f{+}1$ regime achieves two-round optimistic latency~\cite{momose2021}. A recent wave of two-phase $5f{+}1$ protocols, including Alpenglow~\cite{Alpenglow25}, Minimmit~\cite{Minimmit25}, ChonkyBFT~\cite{ChonkyBFT25}, and Kudzu~\cite{Kudzu25}, achieve fast-path finalization at ${\approx}20\%$ Byzantine tolerance and fall back to a slow path resembling $3f{+}1$ for higher fault rates. Our no-fork observation is orthogonal: it is a generic one-round add-on that raises the safety threshold from $20\%$ to $60\%$ in any $5f{+}1$ protocol. Malkhi et al.~\cite{flexibleBFT} introduced the alive-but-corrupt model; this is more restrictive than the deceitful fault model of Basilic~\cite{Basilic23} in that deceitful faults can try to prevent agreement even if that results in no liveness instead. Our no-fork regime tolerates deceitful faults, and by extension alive-but-corrupt faults, without deposits.

\paragraph{Rational consensus.}
Abraham et al.~\cite{ADGH06} formalized $\epsilon$-$(k,t)$-robustness and implemented mediators with cheap talks for $n>k{+}2t$. Abraham et al.~\cite{ADGH19} extended this to asynchronous cheap talks for $n>3(k{+}t)$. Ben-Porath~\cite{BP03} and Heller~\cite{Heller05} studied punishment strategies and coalition-proof equilibria. \Trap{}~\cite{TRAP22} was the first partial-synchrony solution without solution preference, and proved that baiting is necessary when the coalition can fork the finalization phase~\cite[Thm.~3.2]{TRAP22}; we extend this impossibility to the $5f{+}1$ model in Section~\ref{sec:impossibility}.

\paragraph{Accountability and fork analysis.}
Polygraph~\cite{CGG20,CGG21} and ABC~\cite{CGGK22} introduced accountable Byzantine agreement. Sheng et al.~\cite{SWKV21} studied BFT protocol forensics. ZLB~\cite{RG21} and Basilic~\cite{Basilic23} analyzed fork branches as a function of voting thresholds. Basilic achieves resilient-optimal bounds $n>3t{+}d{+}2q$ for consensus and $n>2t{+}d{+}q$ for eventual consensus in the BDB failure model. The fork-branch theorem of ZLB~\cite[Thm.~4.2]{RG21} is the basis of our finalization analysis, and we use the BDB model's eventual consensus mode as a candidate for relaxing the winner consensus bottleneck.

\section{Model}\label{sec:model}

We work in $n=5f{+}1$ under partial synchrony~\cite{DLS88} with $t_0=f$ and quorum $h=4f{+}1$. We consider a game played by a set $N$ of $|N|=n$ players, each of type Byzantine, rational, or correct. At most $t\leq f$ are Byzantine, $k$ are rational, and $n{-}k{-}t$ are correct. The game is in extensive form, described by a game tree, and we introduce a scheduler that models message delays under partial synchrony, alternating moves with the players~\cite{ADGH06,TRAP22}. The coalition is $C=k{+}t$.

Standard cryptography (unforgeable signatures, computationally bounded players) and cheap talks (private pairwise channels with negligible cost) are assumed~\cite{ADGH06,TRAP22}. Each player has a public and private key, and public keys are common knowledge.

\subsection{An extensive form game}\label{sec:gt-model}
Fix a deployment of \Snare{} with parameters $(n,f,\GG,\LL,\RR)$. For every initial proposal vector $\vv{x}=(x_1,\ldots,x_n)$ and type vector
$\vv{\theta}\in\{\mathsf{correct},\mathsf{rational},\mathsf{Byz}\}^n$, the interaction induced by \Snare{} is an extensive form game $\Gamma(\vv{x},\vv{\theta})$ with players $N\cup\{s\}$, where $s$ is the
scheduler. A history is a finite sequence of scheduler moves and player moves. At a scheduler node, $s$ chooses a player $p_i$ to move next and a subset of the messages currently in transit to be delivered to $p_i$
immediately before $p_i$ moves. At a player node, $p_i$ chooses one action allowed by its current local state, sending protocol messages, revealing keys, outputting a value, or stopping. The scheduler is constrained by partial synchrony such that before GST it may delay messages arbitrarily, after GST, every message sent by a non-faulty player is delivered within some finite bound $\Delta$, and every non-terminated player is scheduled infinitely often. The information set of player $p_i$ is its local state, namely its input, private randomness, keys, messages delivered so far, and its own past actions. A behavioral strategy of $p_i$ is a function from information sets to distributions over actions. We write $\SSS_i$ for the strategy set of $p_i$ and
$\SSS_I=\prod_{i\in I}\SSS_i$.

For a terminal history $z$, let $\mathsf{cons}(z)=1$ if and only if all non-faulty players that decide output the same value and that value is valid, otherwise, $\mathsf{cons}(z)=0$. Additionally, let $a(z)\geq 1$ be the number of spendable branches realized in $z$, and let $K(z)\subseteq N$ be the set of rational coalition members who share the loot in $z$, where $\kappa(z)=|K(z)|$. Finally, let $w(z)\in N\cup\{\bot\}$ denote the reward winner selected by the protocol, and let $S(z)\subseteq N$ be the set of players slashed in $z$. Correct players follow the protocol and obtain utility $1$ if
$\mathsf{cons}(z)=1$ and $0$ otherwise. Byzantine players are unrestricted. A rational player $p_i$ has utility
\begin{equation}\label{eq:utility}
\begin{split}
    u_i(z)=&\mathbf{1}[\mathsf{cons}(z)=1]
+\mathbf{1}[a(z)\geq 2 \wedge i\in K(z)]\cdot \frac{(a(z)-1)\GG}{\kappa(z)}\\
+&\mathbf{1}[w(z)=i]\RR
-\mathbf{1}[i\in S(z)]\LL.
\end{split}
\end{equation}
Thus a rational player prefers a successful unique valid decision to
non-termination, but if the coalition realizes an $a$-fold spend then the loot term is added to its payoff.

\subsection{Solution concepts}

We work with $(k,t)$-robustness~\cite{ADGH06}, which strengthens Nash equilibrium to resist joint deviations by a coalition of $k$ rational and $t$ Byzantine players. We restate here the definitions that we require, but refer to \cite{TRAP22} for a detailed explanation. Intuitively, a protocol is $(k,t)$-robust if no rational coalition member can gain by deviating, even when coordinating with up to $t$ Byzantine players. Rational agreement asks for a protocol that is both BFT-safe against $f$ Byzantines alone and $(k,t)$-robust for some coalition exceeding the Byzantine-alone bound. Punishment strategies make deviating costly; baiting strategies go further by rewarding coalition members who expose the coalition, converting the incentive from ``support the fork'' to ``betray for a reward.''

\begin{definition}[$\epsilon$-$(k,t)$-robustness~\cite{ADGH06,TRAP22}]\label{def:robust}
  A joint strategy $\vv{\sigma}\in \mathcal{S}$  is an $\epsilon$-$(k,t)$-robust (resp. strongly $\epsilon$-$(k,t)$-robust) equilibrium
  if for all $K, T\subseteq N$ such that $K\cap T = \emptyset, |K|\leq k,$ and $|T|\leq t$, for all $\vv{\tau}_T \in \mathcal{S}_T$, for all $\vv{\phi}_K\in\mathcal{S}_K$, for some (resp. all)
  $i\in K$, and all strategies of the scheduler $\sigma_s$, we have $u_i(\vv{\sigma}_{-T},\vv{\tau}_T,\sigma_s)\geq u_i(\vv{\sigma}_{N-(K\cup T)}, \vv{\phi}_K, \vv{\tau}_T,\sigma_s)-\epsilon$. We speak instead of a $(k,t)$-robust equilibrium if $\epsilon=0$.
\end{definition}

\begin{definition}[Rational agreement~\cite{TRAP22}]
  Consider a system with $n$ players, a protocol $\vv{\sigma}$ solves
the rational agreement problem if $\vv{\sigma}$ is a $f$-immune
protocol for consensus (i.e. solving consensus against $f$ Byzantines and no rationals), and is also $\epsilon$-$(k,t)$-robust for some
$k>0, t>0$ such that $f+1 \leq (k+t)$ (i.e. solving consensus for some combination that exceeds the Byzantine-alone bound). 
\end{definition}

\begin{definition}[Punishment strategy~\cite{ADGH06,TRAP22}]\label{def:punishment}  A joint strategy $\vv{\gamma}$ is a 
$(k,t)$-punishment strategy with respect to $\vv{\sigma}$ if for all $K,T,P\subseteq N $ such that $K,T,P$ are disjoint, $|K|\leq k,|T|\leq t,|P|> t$, for all $\vv{\tau}\in \mathcal{S}_T$, for all $\vv{\phi}_K\in \mathcal{S}_K$, for all $i\in K$, and all strategies of the scheduler $\sigma_s$, we have $u_i(\vv{\sigma}_{-T},\vv{\tau}_T,\sigma_s) > u_i(\vv{\sigma}_{N-(K\cup T \cup P)}, \vv{\phi}_K, \vv{\tau}_T, \vv{\gamma}_P,\sigma_s)$.
  \end{definition}

\begin{definition}[Baiting strategy~\cite{TRAP22}]\label{def:baiting}
    A joint strategy $\vv{\eta}$ is a $(k,t,m)$-baiting strategy with respect to a strategy $\vv{\sigma}$ if $\vv{\eta}$ is a $(k-m,t)$-punishment strategy with respect to $\vv{\sigma}$, with $0< m \leq k$ and for all $K,T,P\subseteq N$ such that $K\cap T =\emptyset ,\, |P\cap K|\geq m,\, P\cap T = \emptyset$, $|K\backslash P|\leq k-m,|T|\leq t, |P|> t$, for all $\vv{\tau}\in \mathcal{S}_T$, all $\vv{\phi}_{K\backslash P}\in \mathcal{S}_{K\backslash P}-\{\vv{\sigma}_K\}$, all $\vv{\theta}_{P}\in \mathcal{S}_{P}$, all $i\in P$, and all strategies of the scheduler $\sigma_s$, we have $
      u_i(\vv{\sigma}_{N-(K\cup T \cup P)}, \vv{\phi}_{K\backslash P},\vv{\tau}_T, \vv{\eta}_P, \sigma_s)\geq   \\
      u_i(\vv{\sigma}_{N-(K\cup T \cup P)}, \vv{\phi}_{K\backslash P},\vv{\tau}_T,\vv{\theta}_{P},\sigma_s)$.
    Additionally, we speak of a strong $(k,t,m)$-\textit{baiting strategy} in the particular case where for all rational coalitions $K \subseteq N$ such that $|K|\leq k$, $|K\cap P|\geq m$ and all $\vv{\phi}_{K\backslash P}\in \mathcal{S}_{K\backslash P}$ we have:
    $\sum_{i\in K}u_i(\vv{\sigma}_{N-(K\cup P)}, \vv{\phi}_{K\backslash P}, \vv{\eta}_P,\sigma_s) \leq \sum_{i\in K}u_i(\vv{\sigma}, \sigma_s).$
    We write (strong) $(k,t)$-baiting strategy instead to refer to a (strong) $(k,t,m)$-baiting strategy for some $m$, with $0< m \leq k$.
    \end{definition}

We use Definitions~\ref{def:robust}--\ref{def:baiting} relative to the extensive form game $\Gamma(\vv{x},\vv{\theta})$.

\section{The \Snare{} Protocol}\label{sec:snare}
We call the $5f{+}1$ instantiation of \Trap{} the \Snare{} protocol (\emph{Scalable Nash Agreement via Reward and Exclusion}). The $3\times$ quorum intersection in the $5f{+}1$ model amplifies both the detection ($3f{+}1$ fraudsters) and the reward ($\RR=3f\LL$), enabling a wider net than the original $3f{+}1$ formulation. The baiting strategy converts a coalition member's incentive from ``support the disagreement'' to ``betray the coalition for a reward.'' For this to work, the deposit and reward must be calibrated so that baiting is a strictly dominant strategy for a rational coalition member who believes at least $m{-}1$ others also bait (Section~\ref{sec:double}).
We present here the \Snare{} protocol, for which we recall the structure of \Trap{} first.
\subsection{The two layers of \Trap{}}

The \Trap{} architecture has three components: (i)~a financial component (deposit $\LL$ per player, reward $\RR$), (ii)~an accountable consensus layer producing predecisions, and (iii)~a one-shot \BFTCR{} finalization layer producing decisions.

The \BFTCR{} phase, as specified in Algorithm~1 of~\cite{TRAP22}, consists of two reliable broadcasts (RB1 for encrypted commitments, RB2 for lists of $h=4f{+}1$ delivered RB1 messages), a reveal step (broadcasting decryption keys), and either direct decision (if all decrypted hashes agree, line~31 of~\cite{TRAP22}) or a winner consensus with reward/slashing (if proofs of fraud are revealed, lines~36--40 of~\cite{TRAP22}).

Only the output of \BFTCR{} constitutes a decision. Predecision disagreements are the input to the baiting mechanism, not an evasion of it. We note that the full \BFTCR{} commit-reveal mechanism is needed only when the coalition exceeds $3f$ (the fork regimes). Below that threshold, a simpler protocol suffices, as we discuss in Section~\ref{sec:nofork}.

\subsection{Protocol specification}

We present three procedures that compose the \Snare{} protocol.

\paragraph{Simple finalization (Algorithm~\ref{alg:simple}).}
This is the lightweight finalization for the no-fork regime. Each player enters Algorithm~\ref{alg:simple} only after the accountable consensus terminates, which requires a certificate of $4f{+}1$ votes for a single value. A player therefore has a unique predecision $v_i$ backed by such a certificate. If two players hold certificates for different values, more than $f$ players must have signed both certificates, constituting equivocation. Each player broadcasts its signed predecision hash together with its accountable consensus certificate. If $4f{+}1$ matching hashes are collected, the value is decided. If conflicting hashes are found, players fetch the corresponding accountable consensus certificates, extract the equivocating signatures from the conflicting certificates, and resolve the disagreement deterministically. This single all-to-all round is the mechanism that raises safety from ${\approx}20\%$ to ${\approx}60\%$ in any $5f{+}1$ protocol.

\begin{algorithm}[t]
\caption{Simple Finalization for player $p_i$}
\label{alg:simple}
\small
\begin{algorithmic}[1]
\Require predecision $v_i$ and accountable consensus certificate $C_i$
\State \textbf{broadcast} $\langle\textsc{HashCert}, H(v_i), C_i\rangle_{\sigma_i}$ to all
\State $\mathit{certs}[i] \gets C_i$; $\mathit{hashes}[i] \gets H(v_i)$
\Statex \textbf{upon} delivering $\langle\textsc{HashCert}, h, C_j\rangle_{\sigma_j}$ from $p_j$\textbf{:}
\State \hspace{1em}$\mathit{hashes}[j] \gets h$; $\mathit{certs}[j] \gets C_j$
\If{$|\{j : \mathit{hashes}[j] = H(v_i)\}| \geq 4f{+}1$}
  \State \textbf{decide} $v_i$
\EndIf
\If{$\exists\, j,\ell$: $\mathit{hashes}[j] \neq \mathit{hashes}[\ell]$}
  \If{$\textsc{ConflictCert}(\mathit{certs}[j],\mathit{certs}[\ell])$}
    \State $\textit{PoFs} \gets \textsc{ExtractEquivocations}(\mathit{certs}[j],\mathit{certs}[\ell])$
    \State \textbf{slash} players proven in \textit{PoFs}
    \State \textbf{decide} $\textsc{Resolve}(\mathit{certs}[j],\mathit{certs}[\ell])$
  \EndIf
\EndIf
\end{algorithmic}
\end{algorithm}

\paragraph{BFTCR (Algorithm~\ref{alg:bftcr}).}
This is the commit-reveal phase from~\cite{TRAP22}, adapted to $5f{+}1$. It extends simple finalization with encrypted commitments, a two-phase reliable broadcast, and a reveal step that distinguishes baiters from non-baiters. If all revealed hashes agree, the value is decided directly. If proofs of fraud are revealed by a baiter, \BFTCR{} returns the resolved decision together with the PoFs; the financial operations (slashing, winner consensus, reward) are handled by the wrapper (Algorithm~\ref{alg:trap}).

\begin{algorithm}[t]
\caption{\BFTCR{} for player $p_i$ (adapted from~\cite{TRAP22})}
\label{alg:bftcr}
\small
\begin{algorithmic}[1]
\Require predecision $v_i$; key pair $(k_i, k_i^{-1})$
\State $c_i \gets$ \textbf{if} own PoFs \textbf{then} $\textsc{Enc}(k_i, \textit{PoFs})$
\Statex \hspace{3.2em}\textbf{else} $\textsc{Enc}(k_i, H(v_i))$
\State $\textsc{RB1}_i.\textsc{Start}(c_i)$ \Comment{commit encrypted}
\Statex \textbf{upon} RB-delivering $c_j$ from $\textsc{RB1}_j$\textbf{:}
\State \hspace{1em}$\mathit{rb1}[j] \gets c_j$
\State \hspace{1em}\textbf{if} $|\mathit{rb1}| \geq 4f{+}1$ \textbf{then} $\textsc{RB2}_i.\textsc{Start}(\mathit{rb1})$
\Statex \textbf{upon} RB-delivering $L_j$ from $\textsc{RB2}_j$\textbf{:}
\State \hspace{1em}$\mathit{rb2}[j] \gets L_j$
\State \hspace{1em}\textbf{if} $|\mathit{rb2}| \geq 4f{+}1$ \textbf{and} $|\mathit{rb1}| \geq 4f{+}1$
\Statex \hspace{2em}\textbf{then broadcast} $k_i^{-1}$ \Comment{reveal}
\Statex \textbf{upon} delivering $k_j^{-1}$ (and RB1/RB2 from $j$)\textbf{:}
\State \hspace{1em}$d_j \gets \textsc{Dec}(\mathit{rb1}[j],\, k_j^{-1})$
\If{$d_j.\mathit{type} = \textsc{Hash}$ \textbf{and} $4f{+}1$ hashes match}
  \State \textbf{decide} matching value
\ElsIf{$d_j.\mathit{type} = \textsc{PoFs}$ \textbf{and} $\textsc{Verify}(d_j)$}
  \State record valid PoFs from $j$
\EndIf
\If{PoFs received}
  \State \Return $(\textsc{Resolve}(\text{predecisions}), \; \textit{PoFs})$
\EndIf
\end{algorithmic}
\end{algorithm}

\paragraph{\Trap{} wrapper (Algorithm~\ref{alg:trap}).}
The full protocol orchestrates deposits, accountable consensus, and finalization. The choice between simple finalization and \BFTCR{} is a deployment-time configuration reflecting the target regime: if the system is designed for the no-fork regime (safe for any $C\leq 3f$), simple finalization with $\LL=0$ suffices; if it targets the fork regimes (safe for $C$ up to $C_{\textsf{wc}}$), \BFTCR{} with deposits is used. In the fork regime, when \BFTCR{} detects proofs of fraud, the $3f{+}1$ detected equivocators are slashed and excluded, and a winner consensus runs among the $2f$ residual players to select the reward recipient (Lemma~\ref{lem:wc}).
 The actual $(k,t)$ are unknown at deployment; the protocol guarantees safety for whatever $(k,t)$ materializes, as long as the regime's bounds are met. We assume that $\textsc{WinnerConsensus}$ selects a uniformly random valid candidate among the valid candidates, this is the source of the $1/m$ factor in the baiting expected utility calculations below.

\begin{algorithm}[t]
\caption{\Snare{} wrapper for player $p_i$}
\label{alg:trap}
\small
\begin{algorithmic}[1]
\Require proposal $b_i$; deposit $\LL$; regime $\in \{\textsc{NoFork}, \textsc{Fork}\}$
\State \textbf{collect} deposit $\LL$ from each player
\State $v_i \gets \textsc{AccountableConsensus}(b_i)$
\If{regime $= \textsc{NoFork}$} \Comment{system parameter}
  \State $d_i \gets \textsc{SimpleFinalization}(v_i)$ \Comment{Alg.~\ref{alg:simple}}
\Else \Comment{regime $= \textsc{Fork}$}
  \State $(d_i, \textit{PoFs}) \gets \textsc{BFTCR}(v_i)$ \Comment{Alg.~\ref{alg:bftcr}}
\EndIf
\If{fraudsters detected via PoFs}
  \State \textbf{slash} $\LL$ from each proven fraudster
  \State $w \gets \textsc{WinnerConsensus}(\textit{candidates}, \textit{PoFs})$ \Comment{among $2f$ residual}
  \State \textbf{reward} $w$ with $\RR = 3f \cdot \LL$
\EndIf
\State \textbf{return} remaining deposits to honest players
\State \textbf{output} $d_i$
\end{algorithmic}
\end{algorithm}

\section{Finalization Thresholds}\label{sec:branches}

The number of branches (conflicting decided values) a coalition can produce in a single-shot protocol is governed by the same-view partition bound of Ranchal-Pedrosa and Gramoli~\cite[Thm.~4.2]{RG21}. For a protocol with voting threshold $h$ and $n$ players, the minimum coalition for $a$ branches is:
\begin{equation}\label{eq:cfin}
C_{\textsf{fin}}(a)=\left\lceil\frac{ah-n}{a-1}\right\rceil.
\end{equation}
Substituting $n=5f{+}1$ and $h=4f{+}1$, the thresholds for the $5f{+}1$ model are given in Table~\ref{tab:thresholds}.\label{thm:zlb}

\begin{table}[t]
\centering
\caption{Minimum coalition $C_{\textsf{fin}}(a)$ for $a$ branches in the $5f{+}1$ model~\cite{RG21}.}
\label{tab:thresholds}
\small
\begin{tabular}{@{}cccl@{}}
\toprule
$a$ & $C_{\textsf{fin}}(a)$ & $f{=}20$ & Fraction \\
\midrule
$2$ & $3f{+}1$ & $61$ & $60\%$ \\
$3$ & $\lceil(7f{+}2)/2\rceil$ & $71$ & $70\%$ \\
$4$ & $\lceil(11f{+}3)/3\rceil$ & $75$ & $74\%$ \\
$5$ & $\lceil(15f{+}4)/4\rceil$ & $76$ & $75\%$ \\
$\infty$ & $4f{+}1$ & $81$ & $80\%$ \\
\bottomrule
\end{tabular}
\end{table}

For comparison, in $3f{+}1$ ($h=2n/3$): $C_{\textsf{fin}}(2)=n/3\approx 33\%$, $C_{\textsf{fin}}(3)=n/2=50\%$.

As $a\to\infty$, $C_{\textsf{fin}}(a)\to h=4f{+}1$: a coalition controlling $80\%{-}\epsilon$ can split the remaining honest players into arbitrarily many partitions, each reaching quorum with the coalition's help. Note that in any fork with $a\geq 2$ branches, the quorum intersection $2(4f{+}1)-(5f{+}1)=3f{+}1$ ensures that at least $3f{+}1$ equivocating coalition members are eventually detected through their conflicting signatures thanks to accountability. \Trap{} and \Snare{} simply ensure that this disagreement on predecisions is detected before it becomes a disagreement on decisions.\label{prop:detection}

\section{Cross-View Predecisions}\label{sec:predecision}

Before \BFTCR{} starts, the partially synchronous accountable consensus runs across views. A coalition exceeding $f$ can exploit sequential finalization across views: the coalition equivocates in a view change, causing one honest player to predecide one value while the rest predecide another. Each additional predecision consumes at least one honest player who predecides and stops participating. This cross-view disagreement can scale quickly as soon as $k+t>f$.

As a concrete example with $n=5f{+}1$ and coalition $C=f{+}1$: in view~1, all $4f$ honest players and $f{+}1$ coalition members vote for $A$, totaling $5f{+}1$; at least one honest player observes $4f{+}1$ for $A$ and predecides; in view~2, the $f{+}1$ equivocators support $B$, and the scheduler ensures $3f$ honest players (who did not observe $4f{+}1$ for $A$) plus $f{+}1$ equivocators reach $4f{+}1$ for $B$; the remaining honest predecide $B$.

The crucial architectural point is that \BFTCR{} runs \emph{once per decision}. All predecisions from across views are collected as input to a single \BFTCR{} instance. The safety of the final output is therefore governed by the branch bound~\eqref{eq:cfin}~\cite{RG21} applied to this single \BFTCR{} instance, not by the cross-view predecision count. A coalition that creates $f{+}1$ predecisions across views still needs to cause a fork \emph{inside} \BFTCR{} to profit, which requires $C\geq 3f{+}1$ by Table~\ref{tab:thresholds}.

This is why a single all-to-all broadcast round with threshold $4f{+}1$, appended after predecisions and before declaring a decision, already increases safety from ${\approx}20\%$ to ${\approx}60\%$ in any $5f{+}1$ protocol. This round collects predecision hashes from $4f{+}1$ players and decides if they agree. By the branch bound~\eqref{eq:cfin}~\cite{RG21}, no coalition below $3f{+}1$ can cause two disjoint groups to observe different $4f{+}1$-sized quorums in this single round. The cost is one message delay, and it does not block the next view from starting concurrently.

This architecture also provides natural resistance to \emph{long-range attacks}, in which an adversary accumulates equivocations across many past views to retroactively create conflicting histories. Since the finalization round is a single-shot protocol that runs after all view changes have concluded, any cross-view equivocations are collapsed into a single instance where the $4f{+}1$ quorum intersection governs safety. The adversary's accumulated cross-view predecisions are inputs to this round, not evasions of it. The no-fork regime (Theorem~\ref{thm:nofork-safety}) formalizes this: for $C\leq 3f$, the single-shot finalization cannot fork regardless of how many predecisions were created across views.

\section{Analysis}\label{sec:analysis}

We establish three key properties: valid-candidacy (removing the $n>2(k{+}t)$ and $n>\frac{3}{2}k{+}3t$ constraints), winner consensus feasibility (the binding limit), and the no-fork regime (deposit-free safety up to $60\%$). We then prove that baiting is necessary above the no-fork threshold.

\subsection{Valid-candidacy}\label{sec:vc}

In the original \Trap{}, the constraints $n>2(k{+}t)$ and $n>\frac{3}{2}k{+}3t$ (Lemma~4.3 of~\cite{TRAP22}) jointly ensure that baiters' commitments reach enough honest players and that the partition overlap contains enough rational equivocators. We show both are unnecessary. Valid-candidacy requires two properties: (i)~a legitimate baiter is guaranteed to become a valid candidate, and (ii)~a non-baiter cannot become a valid candidate after the baiters reveal. We address both, allowing coalition members to deviate arbitrarily from the reliable broadcast sub-protocol (not only from the consensus protocol), which \Trap{} does not explicitly model.

We introduce one protocol-level modification: a proof-of-baiting (PoB) for candidate $j$ is validated using only RB2 lists from players not in the set of detected equivocators (the $3f{+}1$ players identified through the quorum intersection of the fork partitions). This restriction is well-defined because the set of detected equivocators is determined before PoB validation occurs.

\begin{definition}[Valid candidate]\label{def:valid-candidate}
Let $E$ be the set of detected equivocators in a forked \BFTCR{} execution. A player $p_j$ is a valid candidate if $p_j$ reveals a ciphertext that decrypts to a valid proof of fraud bundle, and there exists a set $Q\subseteq N\setminus E$ with $|Q|\geq f{+}1$ such that,
for every $q\in Q$, the RB2 list broadcast by $p_q$ contains $p_j$'s RB1 commitment.
\end{definition}

\begin{lemma}[Valid-candidacy]\label{lem:vc}
In the $5f{+}1$ \BFTCR{}, with PoB validation restricted to RB2 lists from non-equivocating residual players and $f\geq 3$:
\begin{enumerate}[label=(\roman*),leftmargin=1.5em]
\item Each of the $m(k,t)$ baiters who commits PoFs in RB1 is guaranteed to become a valid candidate upon reveal.
\item No non-baiter can become a valid candidate, regardless of deviations from the reliable broadcast sub-protocol.
\end{enumerate}
\end{lemma}
\begin{proof}
\emph{Part~(i): baiters become valid candidates.}
The baiting strategy (Definition~\ref{def:baiting}) requires the baiter to participate in the fork as an equivocator: this is how the baiter observes conflicting certificates from both fork partitions and constructs PoFs. Since the baiter equivocates like all other coalition members, the fork mechanism delivers their encrypted RB1 commitment to honest players in both partitions. Concretely: each fork partition requires $4f{+}1$ messages to reach quorum, and the coalition delivers all equivocators' RB1 messages to both partitions to achieve this. The coalition cannot selectively suppress the baiter's messages because all RB1 commitments are encrypted: $\textsc{Enc}(k_i, \textit{PoFs})$ is computationally indistinguishable from $\textsc{Enc}(k_i, H(v))$ before the reveal step. All $n{-}C$ honest players therefore deliver the baiter's RB1 and include it in their RB2 lists. These RB2 instances have Bracha agreement (honest senders satisfy $t\leq f < (5f{+}1)/3$). All honest players are in the residual set (they are never equivocators). A PoB requires $f{+}1$ residual RB2 lists containing the commitment. We need $n{-}C\geq f{+}1$, i.e., $C\leq 4f$. Within the winner consensus window $C\leq\lfloor(11f{+}2)/3\rfloor$, this holds for $f\geq 3$ since $(11f{+}2)/3 < 4f$ iff $f>2$. This argument also subsumes the partition-overlap bound $n>\frac{3}{2}k{+}3t$ from Lemma~4.3 of~\cite{TRAP22}, which ensured enough rational equivocators in the partition intersection; here the baiter is guaranteed to be among the equivocators regardless of the partition structure.

\emph{Part~(ii): non-baiters cannot fabricate candidacy.}
After detecting $3f{+}1$ equivocators, the residual set has $n'=2f$ players. The residual coalition has $C'=C{-}(3f{+}1)$ members, and by Lemma~\ref{lem:wc}, $C'<2f/3$. A fabricated PoB for a non-baiter requires $f{+}1$ residual RB2 lists containing a PoF commitment that was never broadcast in RB1. Honest residual players' RB2 lists are consistent (RB agreement for honest senders, $t\leq f<(5f{+}1)/3$) and cannot contain a commitment that was never delivered to them. Only the $C'<2f/3$ residual coalition members' RB2 lists could contain a fabricated entry. Since $2f/3<f{+}1$ for all $f\geq 0$, this is insufficient. Additionally, after $4f{+}1$ RB2 deliveries, at most $n{-}(4f{+}1)=f$ additional RB2 messages can arrive; since $f<f{+}1$, no late commitment can appear in enough residual RB2 lists to form a PoB.
\end{proof}

\begin{remark}\label{rem:vc-trap}
The identical argument applies in $3f{+}1$ ($t_0=f=\lfloor(n{-}1)/3\rfloor$, $h=2f{+}1$). For part~(i): $n{-}C\geq t_0{+}1$ requires $C\leq 2f$, which holds within the winner consensus window $C\leq\lfloor(5f{+}2)/3\rfloor$ for $n\geq 10$. For part~(ii): after detecting $t_0{+}1$ equivocators, the residual coalition has $C'<2f/3<t_0{+}1$ among $n'=2f$ residual players. This removes \emph{both} constraints of Lemma~4.3 of~\cite{TRAP22}: the partition-overlap bound $n>3k/2{+}3t$ and the delivery bound $n>2(k{+}t)$. Both are subsumed by the winner consensus feasibility conditions: $C\leq\lfloor(5f{+}2)/3\rfloor$ and $t\leq f$. In terms of $n$: $n>\max(3t,\;9(k{+}t)/5)$. The original \Trap{} bound $C<n/2$ (from $n>2(k{+}t)$, binding for $k>2t$) extends to $C<5n/9\approx 56\%$, and the bound $C<2n/3{-}t$ (from $n>3k/2{+}3t$, binding for $k<2t$) is similarly relaxed. The constraint $t<2f/3$, previously stated in~\cite{TRAP22}, is also redundant (see Lemma~\ref{lem:wc} proof).
\end{remark}

\subsection{Winner consensus feasibility}\label{sec:wc}

\begin{lemma}[Winner consensus feasibility]\label{lem:wc}
Let $C$ be the coalition size and suppose the \BFTCR{} phase forks into
$a\geq 2$ branches. Then the winner consensus among the residual players solves consensus if
\[
C\leq \left\lfloor\frac{11f+2}{3}\right\rfloor
\qquad\text{and}\qquad
t<\frac{2f}{3}.
\]
\end{lemma}
\begin{proof}
Any fork with $a\geq 2$ detects at least $3f{+}1$ fraudsters, who are excluded. The winner consensus therefore runs on
\[
n'=(5f{+}1)-(3f{+}1)=2f
\]
residual players, among whom the residual coalition size is
\[
C'=C-(3f{+}1).
\]
Safety of a standard partially synchronous BFT protocol on the residual
instance requires $C'<n'/3=2f/3$, i.e.
\[
C<3f{+}1+\frac{2f}{3}=\frac{11f+3}{3}.
\]
Since $C$ is integral, this is equivalent to
\[
C\leq \left\lfloor\frac{11f+2}{3}\right\rfloor.
\]
Moreover, in the worst case none of the detected equivocators is Byzantine, so all $t$ Byzantine players may remain in the residual set; hence we also need $t<2f/3$. Under these inequalities, the residual instance satisfies the usual $n'>3t'$ condition, so safety and liveness follow.
\end{proof}


For $f=20$: $C_{\textsf{wc}}=\lfloor(11\cdot 20{+}2)/3\rfloor=74$ and $t_{\textsf{max}}=f=20$.

\subsection{The no-fork regime}\label{sec:nofork}

\begin{theorem}[Structural no-fork safety]\label{thm:nofork-safety}
If $C\leq 3f$ and $t\leq f$, then every execution of the finalization phase has $a=1$. Equivalently, no two non-faulty players can decide different values.
\end{theorem}
\begin{proof}
For $n=5f{+}1$ and quorum $h=4f{+}1$, the minimum coalition that can create two same view branches is
\[
C_{\textsf{fin}}(2)=2h-n=3f{+}1.
\]
Hence every coalition of size at most $3f$ is below the two-branch threshold, so at most one value can be finalized.
\end{proof}

\begin{theorem}[Deposit free robustness in the no-fork regime]\label{thm:nofork}
Under the utility model of Section~\ref{sec:gt-model}, if the deployment uses \textsc{NoFork} finalization with $\LL=\RR=0$, then the honest profile is a $(k,t)$-robust equilibrium for every $k,t$ such that $k+t\leq 3f$ and $t\leq f$.
\end{theorem}
\begin{proof}
Fix any disjoint $K,T\subseteq N$ with $|K|\leq k$, $|T|\leq t$, any
$\vv{\phi}_K\in\SSS_K$, any $\vv{\tau}_T\in\SSS_T$, and any scheduler strategy $\sigma_s$. Under the honest profile, the accountable consensus layer together with finalization tolerates $t\leq f$ Byzantine players, so every rational player obtains utility $1$.

Now consider the deviating profile $(\vv{\sigma}_{N-(K\cup T)},\vv{\phi}_K,\vv{\tau}_T,\sigma_s)$. By Theorem~\ref{thm:nofork-safety}, every terminal history induced by this profile still satisfies $a=1$. Since $\LL=\RR=0$, the only positive term in \eqref{eq:utility} is the baseline consensus payoff, which is at most $1$. Therefore no rational coalition member can obtain utility greater than $1$ by
deviating. Hence the honest profile is $(k,t)$-robust.
\end{proof}

Theorem~\ref{thm:nofork-safety} is purely structural and does not depend on any
utility function. Theorem~\ref{thm:nofork} is the game-theoretic statement and
uses only the explicit utility model~\eqref{eq:utility}.

\begin{corollary}\label{cor:nofork-rational}
At maximal Byzantine budget $t=f$: $k\leq 2f$, tolerating ${\approx}40\%$ deceitful faults~\cite{Basilic23} (arbitrary unknown utility; the weaker alive-but-corrupt model of~\cite{flexibleBFT} is subsumed) with zero deposits. In total, up to $3f$ ($60\%$) safety-breaking faults, nearly double the $n/3$ ($33\%$) of standard two-round voting, at one extra message delay.
\end{corollary}

\subsection{Impossibility of rational agreement without baiting}\label{sec:impossibility}

\begin{definition}[Solution preference~\cite{TRAP22}]\label{def:solpref}
A protocol has \emph{solution preference} if it designates one proposal as the default outcome. A protocol \emph{without} solution preference treats all proposals symmetrically.
\end{definition}

\begin{definition}[Minimal blocking set]\label{def:min-blocking}
For $C\geq 3f{+}1$, let
\[
m^\star(C):=\left\lfloor\frac{C-(3f{+}1)}{2}\right\rfloor+1.
\]
This is the minimum number of coalition members whose simultaneous defection destroys every two branch fork in the finalization phase.
\end{definition}

\begin{definition}[Punishment only resolution]\label{def:punishment-only}
A fork resolution mechanism is punishment only with cap $P$ if, in every
forking execution each coalition member can lose at most $P$ net utility due to punishment, and exposing the coalition before the second conflicting quorum forms never gives a coalition member a strictly larger transfer than remaining in the coalition.
\end{definition}
Informally, punishment only resolution excludes from the definition of punishment strategies those who are baiting strategies. We prove in the following then that these kind of punishment only strategies are insufficient to solve the rational agreement problem.
\begin{theorem}[Punishment only impossibility above the no-fork threshold]\label{thm:impossibility}
Fix $n=5f{+}1$ and suppose $C=k{+}t\geq 3f{+}1$. Let $\Pi$ be a protocol without solution preference whose fork resolution mechanism is punishment only with cap $P$. If $\Pi$ does not implement a $(k,t,m^\star(C))$-baiting strategy, then for every $\epsilon\geq 0$ there exists $\GG>k(P+1+\epsilon)$ and a scheduler strategy such that the honest profile of $\Pi$ is not $\epsilon$-$(k,t)$-robust.
\end{theorem}
\begin{proof}
Let $h=4f{+}1$. Since $C\geq 3f{+}1=2h-n$, the coalition can create two
conflicting finalization quorums. Take any set $M$ of coalition members with $|M|\leq m^\star(C)-1$ and suppose they defect. Let $Q=C-|M|$ be the remaining coalition members that still support the fork. By the definition of $m^\star(C)$,
\[
2(h-Q)=2(4f{+}1-C+|M|)\leq 5f{+}1-C=n-C.
\]
Hence the $n-C$ non-coalition players can be partitioned into two disjoint sets $A$ and $B$ with $|A|\geq h-Q$ and $|B|\geq h-Q$. The players in $Q$ can equivocate so that $A\cup Q$ reaches quorum for one value and $B\cup Q$ reaches quorum for a conflicting value. Therefore fewer than $m^\star(C)$ defections do not suffice to destroy the fork.

Under the honest profile, every rational player obtains utility $1$ by
consensus. In the deviating profile above, each rational coalition member that supports the fork obtains utility at least $\GG/k - P$, because a two branch fork already yields loot $\GG$ and punishment-only resolution can subtract at most $P$. If $\GG>k(P+1+\epsilon)$, then
\[
\frac{\GG}{k}-P > 1+\epsilon.
\]
Thus some rational coalition member strictly gains more than $\epsilon$ by deviating, contradicting $\epsilon$-$(k,t)$-robustness. Consequently, any $\epsilon$-$(k,t)$-robust protocol above the no-fork threshold must induce at least $m^\star(C)$ coalition members to expose the coalition before the second conflicting quorum forms, i.e.\ it must implement a $(k,t,m^\star(C))$-baiting strategy.
\end{proof}

\section{Double-Spending}\label{sec:double}

For $3f{+}1\leq C<C_{\textsf{fin}}(3)=\lceil(7f{+}2)/2\rceil$, the maximum spending multiplicity is $a=2$, exactly as in the original \Trap{}. The full baiting mechanism is needed.

\subsection{Baiting threshold}

By Definition~\ref{def:min-blocking}, the minimum number of baiters is
\begin{equation}\label{eq:m}
m(k,t)=m^\star(C)=\max\!\left(1,\;\left\lfloor\frac{C-(3f{+}1)}{2}\right\rfloor{+}1\right).
\end{equation}

\subsection{Deposit and reward}

\begin{theorem}[Double-spend regime]\label{thm:double}
For $3f{+}1\leq C<\lceil(7f{+}2)/2\rceil$ with $t\leq f$, the \Trap{} protocol achieves $\epsilon$-$(k,t)$-robustness if each player deposits $\LL=d\cdot\GG$ with
\begin{equation}\label{eq:d2}
d>\frac{m(k,t)}{k\cdot(3f-m(k,t)+1)},
\end{equation}
and the reward is $\RR=3f\cdot\LL$.
\end{theorem}
\begin{proof}
Let $m:=m(k,t)$. Fix any rational coalition member $i$ that contemplates whether to bait or support the fork, assuming that exactly $m-1$ other rational members bait. By Lemma~\ref{lem:vc}(i), every baiter becomes a valid candidate. By Lemma~\ref{lem:vc}(ii), no non-baiter can fabricate candidacy. Since $\textsc{WinnerConsensus}$ selects uniformly among the valid candidates, a baiter wins the reward with probability exactly $1/m$. Hence the expected baiting payoff of $i$ is
\[
U_i^{\mathsf{bait}}
=\frac{1}{m}\RR-\frac{m-1}{m}\LL
=\frac{3f-m+1}{m}\LL.
\]

If instead $i$ supports the fork, then in the double spend regime the maximum loot available to any rational coalition member is
\[
U_i^{\mathsf{fork}}\leq \frac{\GG}{k}.
\]
Therefore condition~\eqref{eq:d2} implies
\[
U_i^{\mathsf{bait}} > U_i^{\mathsf{fork}},
\]
so baiting is a strict best response once $m-1$ others bait.

Next, because
\[
m > \frac{C-(3f+1)}{2},
\]
removing the $m$ baiters from both branches leaves at most $C-m$ coalition members supporting the fork, which is insufficient to complete two quorums of size $4f+1$, this is exactly the blocking threshold of Definition~\ref{def:min-blocking}. Hence the fork cannot be finalized once the $m$ baiters defect.

Finally, by Lemma~\ref{lem:wc}, the residual winner consensus is feasible throughout the present regime because
\[
C<\left\lceil\frac{7f+2}{2}\right\rceil
\leq \left\lfloor\frac{11f+2}{3}\right\rfloor
\qquad\text{and}\qquad
t<\frac{2f}{3}.
\]
The reward is loss free because at least $3f+1$ equivocators are slashed, so their deposits cover $\RR=3f\LL$. Moreover,
\[
\RR-(C-1)\LL=(3f-C+1)\LL\le 0
\]
for every $C\ge 3f+1$, so the coalition cannot profit by self-triggering the reward. Therefore the profile in which $m$ rational coalition members bait is a strong $(k,t,m)$-baiting strategy, and the induced profile is $\epsilon$-$(k,t)$-robust.
\end{proof}

\subsection{Deposit table and concrete example}

Table~\ref{tab:double} gives deposit values for $n=101$ ($f=20$, $t_{\max}=20$).

\begin{table}[t]
\centering
\caption{Double-spend regime ($a=2$), $n=101$, $f=20$, $t=20$ (worst case).}
\label{tab:double}
\small
\begin{tabular}{@{}ccccc@{}}
\toprule
$C$ & $k$ & $m$ & $d_{\min}$ & $\LL/\GG$ \\
\midrule
$61$ & $41$ & $1$ & $0.00041$ & $0.041\%$ \\
$63$ & $43$ & $2$ & $0.00079$ & $0.079\%$ \\
$65$ & $45$ & $3$ & $0.0011$ & $0.11\%$ \\
$67$ & $47$ & $4$ & $0.0015$ & $0.15\%$ \\
$70$ & $50$ & $5$ & $0.0018$ & $0.18\%$ \\
\bottomrule
\end{tabular}
\end{table}

Deposits remain under $0.2\%$ of the gain, roughly $3\times$ lower than a na\"ive analysis using only $f{+}1$ detected fraudsters. The improvement comes from the $3f{+}1$ quorum intersection: the reward $\RR=3f\LL$ is $3\times$ larger than the $f\LL$ that would follow from $f{+}1$ detections, so a smaller deposit suffices to make baiting dominant. Concretely: consider a blockchain with 101 validators processing blocks worth up to \$10M each. With $C=65$ ($20$ Byzantine, $45$ rational), each validator deposits $d\cdot\GG\approx 0.115\%\times\$10\text{M}\approx\$11{,}500$, for a total of approximately \$1.2M across all validators, securing against a $65\%$ coalition attempting a double-spend.

\subsection{The $k$-$t$ tradeoff}

The double-spend regime exhibits a tradeoff between rational and Byzantine tolerance that mirrors the original \Trap{}~\cite{TRAP22}. Combining the no-fork and double-spend regimes into a single bound:

\begin{corollary}[$k$-$t$ tradeoff for double-spend tolerance]\label{cor:kt-tradeoff}
\Snare{} tolerates at most double-spending if:
\begin{equation}\label{eq:kt-bound}
n > \max\!\left(5t,\;\; \frac{10(k{+}t)}{7}\right).
\end{equation}
The first constraint ($5t$) comes from $t\leq f=(n{-}1)/5$. The second ($10(k{+}t)/7$) comes from the double-spend finalization threshold ($C < \lceil(7f{+}2)/2\rceil$). The no-fork regime bound $n>\max(5t, 5(k{+}t)/3)$ is subsumed: $5(k{+}t)/3 < 10(k{+}t)/7$ since $5/3>10/7$.
\end{corollary}
\begin{proof}
The double-spend regime requires $C<\lceil(7f{+}2)/2\rceil$ and $t\leq f$. With $f=(n{-}1)/5$: $k{+}t<7(n{-}1)/10$, i.e., $n>10(k{+}t)/7{+}1$; and $t\leq (n{-}1)/5$, i.e., $n\geq 5t{+}1$. The constraint $t<2f/3$, previously required for the winner consensus Byzantine budget, is redundant (see Lemma~\ref{lem:wc}).
\end{proof}

For comparison, the original $3f{+}1$ \Trap{} has $n>\max(\frac{3}{2}k{+}3t,\;2(k{+}t))$; with the retroactive valid-candidacy improvement (Remark~\ref{rem:vc-trap}), this relaxes to $n>\max(3t,\;9(k{+}t)/5)$. Table~\ref{tab:kt} compares the tradeoffs at specific operating points.

\begin{table}[t]
\centering
\caption{$k$-$t$ tradeoff at $n=101$: maximum rational and Byzantine fractions per regime.}
\label{tab:kt}
\small
\begin{tabular}{@{}lcccc@{}}
\toprule
\textbf{Model} & $t_{\max}$ & $k$ at $t_{\max}$ & $t$ at $k_{\max}$ & $k_{\max}$ \\
\midrule
$3f{+}1$~\cite{TRAP22} & $33$ & $1$ & $1$ & $49$ \\
$5f{+}1$ no-fork & $20$ & $40$ & $1$ & $59$ \\
$5f{+}1$ dbl-spend & $20$ & $50$ & $1$ & $69$ \\
$5f{+}1$ trpl-spend & $20$ & $54$ & $1$ & $73$ \\
\bottomrule
\end{tabular}
\end{table}

\section{Triple-Spending and the Winner Consensus Limit}\label{sec:triple}

For $\lceil(7f{+}2)/2\rceil\leq C\leq C_{\textsf{wc}}$, the maximum multiplicity is $a=3$.

\begin{theorem}[Triple-spend regime]\label{thm:triple}
For $\lceil(7f{+}2)/2\rceil\leq C\leq\lfloor(11f{+}2)/3\rfloor$ with $t\leq f$, \Trap{} achieves $\epsilon$-$(k,t)$-robustness if each player deposits $\LL=d\cdot\GG$ with
\begin{equation}\label{eq:d3}
d>\frac{2\cdot m(k,t)}{k\cdot(3f-m(k,t)+1)},
\end{equation}
and the reward is $\RR=3f\cdot\LL$.
\end{theorem}

\begin{proof}
Let $m:=m(k,t)$. Fix any rational coalition member $i$ that contemplates whether to bait or support the fork, assuming that exactly $m-1$ other rational members bait. By Lemma~\ref{lem:vc}(i), every baiter becomes a valid candidate. By Lemma~\ref{lem:vc}(ii), no non-baiter can fabricate candidacy. Since $\textsc{WinnerConsensus}$ selects uniformly among the valid candidates, a baiter wins the reward with probability exactly $1/m$. Hence the expected baiting payoff of $i$ is
\[
U_i^{\mathsf{bait}}
=\frac{1}{m}\RR-\frac{m-1}{m}\LL
=\frac{3f-m+1}{m}\LL.
\]

If instead $i$ supports the fork, then in the triple spend regime the maximum loot available to any rational coalition member is
\[
U_i^{\mathsf{fork}}\leq \frac{2\GG}{k}.
\]
Therefore condition~\eqref{eq:d3} implies
\[
U_i^{\mathsf{bait}} > U_i^{\mathsf{fork}},
\]
so baiting is a strict best response once $m-1$ others bait.

Next, because
\[
m > \frac{C-(3f+1)}{2},
\]
removing the $m$ baiters from both branches leaves at most $C-m$ coalition members supporting the fork, which is insufficient to complete two quorums of size $4f+1$. Hence the fork cannot be finalized once the $m$ baiters defect.

Finally, by Lemma~\ref{lem:wc}, the residual winner consensus is feasible throughout the present regime because
\[
C\leq \left\lfloor\frac{11f+2}{3}\right\rfloor
\qquad\text{and}\qquad
t<\frac{2f}{3}.
\]
The reward is loss free because at least $3f+1$ equivocators are slashed, so their deposits cover $\RR=3f\LL$. Moreover,
\[
\RR-(C-1)\LL=(3f-C+1)\LL<0
\]
throughout the present regime, so the coalition cannot profit by self-triggering the reward.

Therefore the profile in which $m$ rational coalition members bait is a strong $(k,t,m)$-baiting strategy, and the induced profile is $\epsilon$-$(k,t)$-robust.
\end{proof}

\begin{table}[t]
\centering
\caption{Triple-spend regime ($a=3$), $n=101$, $f=20$, $t=20$ (worst case).}
\label{tab:triple}
\small
\begin{tabular}{@{}ccccc@{}}
\toprule
$C$ & $k$ & $m$ & $d_{\min}$ & $\LL/\GG$ \\
\midrule
$71$ & $51$ & $6$ & $0.0043$ & $0.43\%$ \\
$72$ & $52$ & $6$ & $0.0042$ & $0.42\%$ \\
$73$ & $53$ & $7$ & $0.0049$ & $0.49\%$ \\
$74$ & $54$ & $7$ & $0.0048$ & $0.48\%$ \\
\bottomrule
\end{tabular}
\end{table}

In the same \$10M-block scenario: at $C=74$ ($20$ Byzantine, $54$ rational), each validator deposits $\approx 0.48\%\times\$10\text{M}=\$48{,}000$, totaling \$4.8M staked, still modest relative to the \$10M secured.

\subsection{Quadruple-spending and the winner consensus ceiling}

The threshold for four branches is $C_{\textsf{fin}}(4)=\lceil(11f{+}3)/3\rceil$. The winner consensus limit is $C_{\textsf{wc}}=\lfloor(11f{+}2)/3\rfloor$. For $f=20$: $C_{\textsf{fin}}(4)=75$ while $C_{\textsf{wc}}=74$, so quadruple-spending exceeds the winner consensus window by $1$ player. A coalition of $75$ \emph{can} create four branches inside \BFTCR{}, but the resolution mechanism breaks: after detecting $3f{+}1=61$ fraudsters, $2f=40$ players remain with $14$ residual coalition members, exceeding $40/3\approx 13.3$. This is not a limitation of the branch formula, which permits more branches for larger coalitions, but of the winner consensus being a standard partial-synchrony BFT consensus with the $n'>3t'$ requirement~\cite{DLS88}.

At $C=75$, $a=4$, the total coalition-certificate appearances are $\geq a(4f{+}1)-(n{-}C)=4\cdot 81-26=298$. Each coalition member appears in at most $a=4$ certificates. If $x$ members appear in exactly one certificate (undetected), the total is $\leq x+4(75-x)=300-3x$. Thus $298\leq 300-3x$, giving $x\leq 0$: every single coalition member appears in $\geq 2$ certificates and is detected. At $C=75$ with $a=4$, detection is total and the WC runs on 26 honest players with zero faults. The rational response is to create only $a=2$ branches (minimizing detection at $3f{+}1=61$), which is what makes the WC bottleneck bind in practice. This covering argument generalizes: for any $a\geq 2$, the number of undetected coalition members is $x\leq\lfloor(a\cdot n-a\cdot h-n+C)/(a-1)\rfloor$ where $h=4f{+}1$.

\section{Summary of Regimes}\label{sec:summary}

Figure~\ref{fig:regimes} and Tables~\ref{tab:regimes}--\ref{tab:comparison} summarize the regime structure. Figure~\ref{fig:branches} compares the branch thresholds $C_{\textsf{fin}}(a)/n$ across the two models, showing the substantially higher coalition required for each branch count in $5f{+}1$. Figure~\ref{fig:feasibility} plots the feasible $(t/n, k/n)$ pairs for each regime: the $5f{+}1$ no-fork zone (green) already exceeds the full $3f{+}1$ \Trap{} for $t\leq 10\%$, and the fork regimes (blue) extend to over $73\%$ total coalition. Figure~\ref{fig:deposits} shows the deposit requirements as a function of $C$ for $n=101$.

\begin{figure*}[t]
\centering
\begin{tikzpicture}[xscale=0.17,yscale=0.7]
\node[font=\small\bfseries,anchor=east] at (-2,0.5) {$5f{+}1$};
\fill[blue!12] (0,0) rectangle (20,1);
\fill[green!18] (20,0) rectangle (60,1);
\fill[orange!18] (60,0) rectangle (70,1);
\fill[red!14] (70,0) rectangle (74,1);
\fill[gray!14] (74,0) rectangle (100,1);
\draw[thick] (0,0) rectangle (100,1);
\foreach \x in {20,60,70,74} {\draw[thick] (\x,0) -- (\x,1);}
\node[font=\footnotesize\bfseries] at (10,0.5) {BFT};
\node[font=\footnotesize\bfseries] at (40,0.5) {No-fork ($\LL{=}0$)};
\node[font=\footnotesize\bfseries] at (65,0.5) {$a{=}2$};
\node[font=\scriptsize\bfseries] at (71.5,0.5) {$3$};
\node[font=\scriptsize,gray] at (86,0.5) {beyond WC};
\foreach \x/\lab in {0/0\%,20/20\%,60/60\%,70/70\%,74/74\%,80/80\%,100/100\%} {
  \node[font=\scriptsize,above,inner sep=1pt] at (\x,1.05) {\lab};
}
\node[font=\small\bfseries,anchor=east] at (-2,-1.5) {$3f{+}1$};
\fill[blue!12] (0,-2) rectangle (33,-1);
\fill[orange!18] (33,-2) rectangle (50,-1);
\fill[red!14] (50,-2) rectangle (56,-1);
\fill[gray!14] (56,-2) rectangle (100,-1);
\draw[thick] (0,-2) rectangle (100,-1);
\foreach \x in {33,50,56} {\draw[thick] (\x,-2) -- (\x,-1);}
\node[font=\footnotesize\bfseries] at (16.5,-1.5) {BFT};
\node[font=\footnotesize\bfseries] at (41.5,-1.5) {$a{=}2$};
\node[font=\scriptsize\bfseries] at (53,-1.5) {$3$};
\node[font=\scriptsize,gray] at (78,-1.5) {beyond WC};
\foreach \x/\lab in {0/0\%,33/33\%,50/50\%,56/56\%,67/67\%,100/100\%} {
  \node[font=\scriptsize,below,inner sep=1pt] at (\x,-2.05) {\lab};
}
\node[font=\small,below] at (50,-2.65) {Coalition size $C/n$};
\end{tikzpicture}
\caption{Regime comparison. \textbf{Top:} \Snare{} (this work). The no-fork regime ($\LL{=}0$) spans $20\%$ to $60\%$; the fork regimes extend to $74\%$ (WC limit). \textbf{Bottom:} $3f{+}1$ \Trap{}~\cite{TRAP22} with the retroactive valid-candidacy fix (Lemma~\ref{lem:vc}). No no-fork regime exists (predecision and finalization thresholds coincide at $33\%$); the original bound was $50\%$, extended to $56\%$ by the VC fix. Gray zones require a more fault-tolerant winner selection mechanism.}
\label{fig:regimes}
\end{figure*}
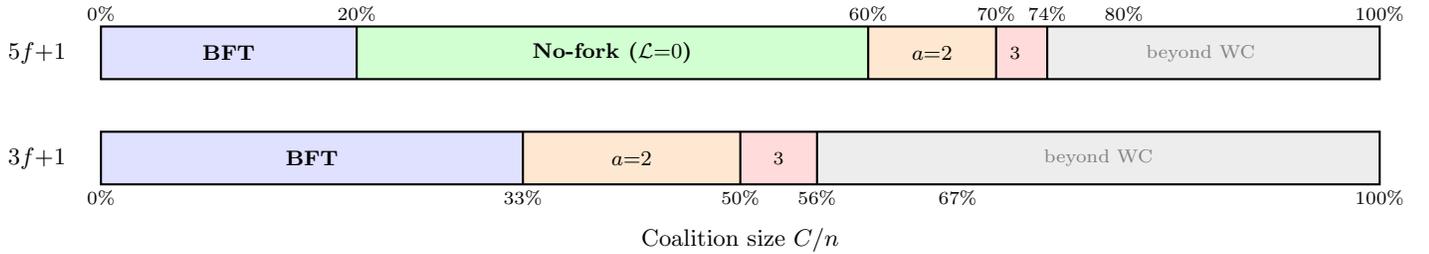

\begin{figure}[t]
\centering
\begin{tikzpicture}
\begin{axis}[
  width=\columnwidth, height=5.2cm,
  xlabel={Branches $a$},
  ylabel={Min.\ coalition $C_{\textsf{fin}}(a)/n$},
  xmin=1.5, xmax=12, ymin=0.2, ymax=0.85,
  xtick={2,3,4,5,6,8,10,12},
  ytick={0.2,0.3,0.4,0.5,0.6,0.7,0.8},
  yticklabel={\pgfmathprintnumber\tick},
  grid=major,
  grid style={gray!25},
  legend style={at={(0.02,0.98)},anchor=north west,font=\scriptsize,draw=none,fill=white,fill opacity=0.85},
  every axis label/.style={font=\small},
  tick label style={font=\scriptsize},
]
\addplot[thick, blue, mark=*, mark size=1.5pt, domain=2:12, samples=11]
  {(4*x - 5)/(5*(x-1))};
\addlegendentry{$5f{+}1$ ($h{=}4n/5$)}

\addplot[thick, red, mark=square*, mark size=1.5pt, domain=2:12, samples=11]
  {(2*x - 3)/(3*(x-1))};
\addlegendentry{$3f{+}1$ ($h{=}2n/3$)}

\draw[blue, dashed, thin] (axis cs:1.5,0.8) -- (axis cs:12,0.8);
\draw[red, dashed, thin] (axis cs:1.5,0.6667) -- (axis cs:12,0.6667);

\draw[blue, dotted, thick] (axis cs:1.5,0.733) -- (axis cs:12,0.733);
\node[font=\tiny,blue,fill=white,inner sep=1pt] at (axis cs:9.5,0.755) {$5f{+}1$ WC};

\draw[red, dotted, thick] (axis cs:1.5,0.556) -- (axis cs:12,0.556);
\node[font=\tiny,red,fill=white,inner sep=1pt] at (axis cs:9.5,0.578) {$3f{+}1$ WC (new)};

\draw[red, dashdotted, thin] (axis cs:1.5,0.5) -- (axis cs:12,0.5);
\node[font=\tiny,red!60,fill=white,inner sep=1pt] at (axis cs:4,0.48) {$3f{+}1$ original};

\end{axis}
\end{tikzpicture}
\caption{Minimum coalition $C_{\textsf{fin}}(a)/n$ for $a$ branches. Dashed: quorum thresholds. Dotted: winner consensus (WC) limits. The $3f{+}1$ WC limit at $5n/9\approx 56\%$ is the retroactive improvement from Lemma~\ref{lem:vc}; the original bound was $n/2$ (dash-dotted).}
\label{fig:branches}
\end{figure}
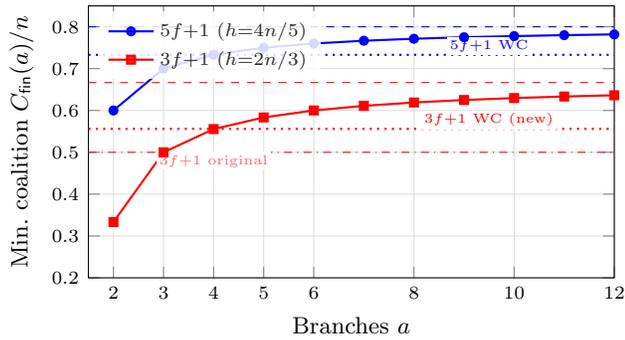

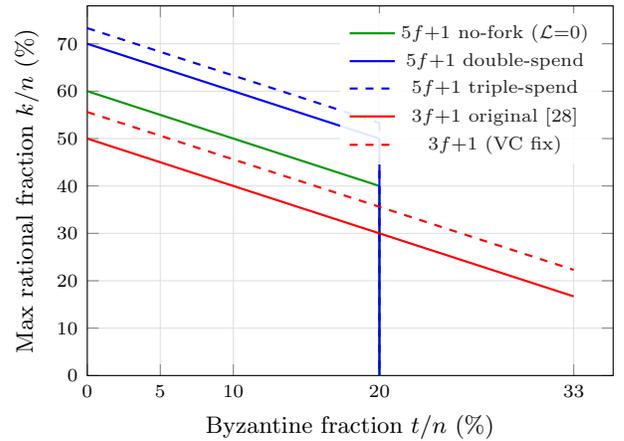
\begin{figure}[t]
\centering
\begin{tikzpicture}
\begin{axis}[
  width=\columnwidth, height=6.5cm,
  xlabel={Byzantine fraction $t/n$ (\%)},
  ylabel={Max rational fraction $k/n$ (\%)},
  xmin=0, xmax=0.36, ymin=0, ymax=0.78,
  xtick={0,0.05,0.10,0.20,0.333},
  xticklabels={$0$,$5$,$10$,$20$,$33$},
  ytick={0,0.1,0.2,0.3,0.4,0.5,0.6,0.7},
  yticklabels={$0$,$10$,$20$,$30$,$40$,$50$,$60$,$70$},
  grid=major,
  grid style={gray!25},
  legend style={at={(0.98,0.98)},anchor=north east,font=\scriptsize,draw=none,fill=white,fill opacity=0.9},
  every axis label/.style={font=\small},
  tick label style={font=\scriptsize},
]

\addplot[thick,green!60!black,domain=0:0.2,samples=50]
  {0.6 - x};
\addplot[thick,green!60!black,forget plot] coordinates {(0.2,0.4)(0.2,0)};
\addlegendentry{$5f{+}1$ no-fork ($\LL{=}0$)}

\addplot[thick,blue,domain=0:0.2,samples=50]
  {0.7 - x};
\addplot[thick,blue,forget plot] coordinates {(0.2,0.5)(0.2,0)};
\addlegendentry{$5f{+}1$ double-spend}

\addplot[thick,blue,dashed,domain=0:0.2,samples=50]
  {0.733 - x};
\addplot[thick,blue,dashed,forget plot] coordinates {(0.2,0.533)(0.2,0)};
\addlegendentry{$5f{+}1$ triple-spend}

\addplot[thick,red,domain=0:0.333,samples=50]
  {max(0.5 - x, 0)};
\addlegendentry{$3f{+}1$ original~\cite{TRAP22}}

\addplot[thick,red,dashed,domain=0:0.333,samples=50]
  {max(0.556 - x, 0)};
\addlegendentry{$3f{+}1$ (VC fix)}

\end{axis}
\end{tikzpicture}
\caption{Feasible $(t/n, k/n)$ pairs. The $5f{+}1$ no-fork regime (green, $\LL{=}0$) covers $t$ up to $20\%$ for coalitions up to $60\%$. The fork regimes (blue) extend beyond $60\%$ with the same $t\leq 20\%$ Byzantine tolerance as the no-fork regime: the constraint $t<2f/3$ previously stated is redundant because the residual Byzantine count is bounded by $C'<2f/3$ regardless of $t$ (see Lemma~\ref{lem:wc}). The $3f{+}1$ model (red) tolerates higher $t$ but lower total coalition.}
\label{fig:feasibility}
\end{figure}

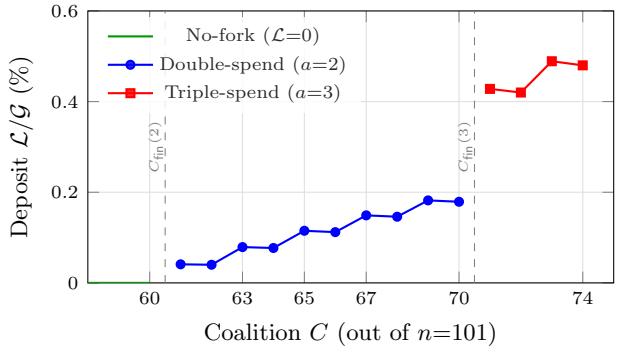
\begin{figure}[t]
\centering
\begin{tikzpicture}
\begin{axis}[
  width=\columnwidth, height=5.2cm,
  xlabel={Coalition $C$ (out of $n{=}101$)},
  ylabel={Deposit $\LL/\GG$ (\%)},
  xmin=58, xmax=75, ymin=0, ymax=0.6,
  xtick={60,63,65,67,70,74},
  grid=major,
  grid style={gray!25},
  legend style={at={(0.02,0.98)},anchor=north west,font=\scriptsize,draw=none,fill=white,fill opacity=0.85},
  every axis label/.style={font=\small},
  tick label style={font=\scriptsize},
]

\addplot[thick,green!60!black,mark=none] coordinates {(58,0)(60,0)};
\addlegendentry{No-fork ($\LL{=}0$)}

\addplot[thick,blue,mark=*,mark size=1.5pt] coordinates {
(61, 0.041)
(62, 0.040)
(63, 0.079)
(64, 0.077)
(65, 0.115)
(66, 0.112)
(67, 0.149)
(68, 0.146)
(69, 0.182)
(70, 0.179)
};
\addlegendentry{Double-spend ($a{=}2$)}

\addplot[thick,red,mark=square*,mark size=1.5pt] coordinates {
(71, 0.428)
(72, 0.420)
(73, 0.489)
(74, 0.480)
};
\addlegendentry{Triple-spend ($a{=}3$)}

\draw[gray,dashed] (axis cs:60.5,0) -- (axis cs:60.5,0.6);
\draw[gray,dashed] (axis cs:70.5,0) -- (axis cs:70.5,0.6);
\node[font=\tiny,gray,rotate=90] at (axis cs:60.2,0.3) {$C_{\textsf{fin}}(2)$};
\node[font=\tiny,gray,rotate=90] at (axis cs:70.2,0.3) {$C_{\textsf{fin}}(3)$};

\end{axis}
\end{tikzpicture}
\caption{Required deposit as a fraction of the maximum gain per block for $n{=}101$, $f{=}20$, $t{=}20$ (worst case). The step pattern reflects the discrete baiting threshold $m$. The jump at $C_{\textsf{fin}}(3){=}71$ is due to both the increased gain factor ($2\GG$ vs $\GG$) and the larger $m$.}
\label{fig:deposits}
\end{figure}

The regimes form a natural progression. Below $20\%$, standard BFT consensus guarantees safety directly. From $20\%$ to $60\%$, the accountable consensus can be attacked across views but the one-shot finalization round prevents any final fork: this is the no-fork zone, requiring zero deposits. From $60\%$ to $70\%$, double-spending becomes possible inside the finalization phase, and the full \Trap{} baiting mechanism with tiny deposits ($<0.2\%$) ensures robustness. From $70\%$ to $74\%$, triple-spending becomes possible with deposits under $0.5\%$. Beyond $74\%$, the winner consensus cannot tolerate the residual coalition.

For a system with 101 validators and \$10M blocks: in the no-fork regime, zero capital is locked. In the double-spend regime (e.g., $C=65$), total staked capital is \$1.2M to secure \$10M per block. In the triple-spend regime ($C=74$), total staked capital is \$4.8M. These are modest relative to the value secured, especially considering that deposits are returned after the finalization window. The key driver of these low deposits is the $3f{+}1$ quorum intersection: because so many fraudsters are detected, the reward $\RR=3f\LL$ can be $3\times$ larger than a na\"ive $f\LL$ analysis would suggest, making baiting dominant at a correspondingly lower deposit.

\begin{table}[t]
\centering
\caption{The four regimes for $n=101$ ($f=20$).}
\label{tab:regimes}
\small
\begin{tabular}{@{}lccccc@{}}
\toprule
\textbf{Regime} & $C$ & $a_{\max}$ & $t_{\max}$ & $k_{\max}$ & $d_{\max}$ \\
\midrule
BFT & ${\leq}20$ & $1$ & $20$ & $0$ & $0$ \\
No-fork & $21$--$60$ & $1$ & $20$ & $40$ & $0$ \\
Double-spend & $61$--$70$ & $2$ & $20$ & $50$ & $0.18\%$ \\
Triple-spend & $71$--$74$ & $3$ & $20$ & $54$ & $0.49\%$ \\
\bottomrule
\end{tabular}
\end{table}

\begin{table}[t]
\centering
\caption{Comparison with $3f{+}1$ \Trap{}~\cite{TRAP22}.}
\label{tab:comparison}
\small
\begin{tabular}{@{}lcc@{}}
\toprule
& \textbf{$3f{+}1$~\cite{TRAP22}} & \textbf{\Snare{} (this work)} \\
\midrule
Predec.\ threshold & $C\geq n/3$ & $C\geq n/5$ \\
Final $a{=}2$ & $C\geq n/3$ & $C\geq 3n/5$ \\
Final $a{=}3$ & $C\geq n/2$ & $C\geq 7n/10$ \\
Predec./final gap & $1\times$ & $3\times$ \\
No-fork max $C$ & n/a & $3n/5$ \\
Full max $C$ & $n/2$ & $\approx 11n/15$ \\
Full max $C$ (VC fix) & $5n/9$ & $\approx 11n/15$ \\
Max $k$ ($t{=}t_{\max}$) & $n/6$ & $\approx 8n/15$ \\
Deposit (no-fork) & n/a & $0$ \\
Max deposit ($a{=}2$) & $\approx 3\%$ & ${<}0.2\%$ \\
\bottomrule
\end{tabular}
\end{table}

Table~\ref{tab:comparison} highlights the structural improvements. In $3f{+}1$, the predecision and finalization thresholds for $a=2$ coincide at $n/3$, so the layered structure is invisible: every predecision disagreement immediately threatens a final fork. In $5f{+}1$, they diverge by a factor of $3$ ($n/5$ vs $3n/5$), creating the large no-fork zone. The ``VC fix'' row shows the retroactive improvement from Lemma~\ref{lem:vc}: in $3f{+}1$, the original bound $n/2$ could have been $5n/9$. In $5f{+}1$, this fix has no additional effect because the winner consensus (Lemma~\ref{lem:wc}) is already the binding constraint.

\section{Correctness}\label{sec:correctness}

\begin{theorem}[Main]\label{thm:main}
The \Snare{} protocol solves rational agreement in all regimes:
\begin{enumerate}[label=(\roman*),leftmargin=1.5em]
\item $C\leq f$, $t\leq f$: standard BFT safety.
\item $f<C\leq 3f$, $t\leq f$: Theorem~\ref{thm:nofork}.
\item $3f{+}1\leq C<\lceil(7f{+}2)/2\rceil$, $t\leq f$: Theorem~\ref{thm:double}.
\item $\lceil(7f{+}2)/2\rceil\leq C\leq\lfloor(11f{+}2)/3\rfloor$, $t\leq f$: Theorem~\ref{thm:triple}.
\end{enumerate}
\end{theorem}
\begin{proof}
The $t_0$-immunity in all regimes follows from the accountable consensus tolerating $t\leq f$ and the finalization round's reliable broadcasts terminating with $\geq 4f{+}1$ participants.

Termination: in the no-fork regime, the all-to-all finalization round terminates after GST since all honest players broadcast and deliver within $\Delta$. In the fork regimes, the \BFTCR{} phase terminates by the liveness of reliable broadcast and the reveal step, and the winner consensus terminates by standard partial-synchrony liveness since the residual coalition satisfies $C'<n'/3$ (Lemma~\ref{lem:wc}).

For $\epsilon$-$(k,t)$-robustness: regime (i) is standard BFT safety. Regime (ii) follows from Theorems~\ref{thm:nofork-safety} and~\ref{thm:nofork}, Theorem~\ref{thm:nofork-safety} gives the structural no-fork property, and Theorem~\ref{thm:nofork} lifts it to $(k,t)$-robustness under the explicit utility model of Section~\ref{sec:gt-model}. Regimes (iii) and (iv) follow from Theorems~\ref{thm:double} and~\ref{thm:triple} respectively: valid-candidacy (Lemma~\ref{lem:vc}) ensures each baiter becomes a valid candidate~(i) and no non-baiter can fabricate candidacy~(ii); baiting-dominance ensures at least $m(k,t)$ rational players betray the coalition; the winner consensus (Lemma~\ref{lem:wc}) selects the reward recipient; lossfree-reward ensures the reward is funded by slashed deposits; and the strong baiting property prevents coalitions from self-triggering the reward. In all fork cases, the $m$ baiters prevent finalization of the disagreement, and the winner consensus resolves it deterministically.
\end{proof}

Liveness of the accountable consensus layer is orthogonal to the \BFTCR{} mechanism: any partially synchronous BFT protocol tolerating $t\leq f$ Byzantine faults provides liveness after GST, and the \BFTCR{} phase does not modify the consensus layer. The finalization round, reliable broadcasts, and reveal step all terminate after GST by construction, even if rationals that know they will be caught and slashed via PoFs stop becoming live, due to them being dynamically removed from the committee and thresholds updated, in what is called active accountability by prior work \cite{Basilic23}. Liveness of the overall protocol is therefore inherited from the underlying accountable consensus, and the game-theoretic analysis (safety via $(k,t)$-robustness) is independent of the liveness argument.

\section{Impact on Prior Work}\label{sec:impact}

\paragraph{A protocol-agnostic safety amplification.}
The observation that one extra all-to-all broadcast round with threshold $4f{+}1$ increases safety from ${\approx}20\%$ to ${\approx}60\%$ is independent of \Trap{}. It applies to any $5f{+}1$ consensus protocol: append the round after the main consensus produces a predecision, before declaring it final. The round is off the critical path (subsequent views proceed concurrently) and adds one message delay total, not one per view. Combining this with mixed-model approaches such as ebb-and-flow protocols~\cite{ENT21}, where liveness comes from an available chain layer and finality from a BFT gadget, is a natural direction for future work: the $5f{+}1$ finality gadget would enjoy the $3\times$ safety amplification while liveness is handled by the available layer.

\paragraph{ZLB~\cite{RG21}.}
The fork-branch bound~\eqref{eq:cfin}~\cite{RG21} remains correct as a same-view partition result. Cross-view attacks create predecision disagreements below the same-view threshold. In $3f{+}1$ the distinction is invisible ($C_{\textsf{fin}}(2)=n/3$ coincides with the predecision threshold); in $5f{+}1$ they diverge $3\times$.

\paragraph{\Trap{}~\cite{TRAP22}.}
Theorem~4.2 of~\cite{TRAP22} explicitly handles conflicting predecisions and shows \BFTCR{} resolves them, which is the right abstraction. Lemma~\ref{lem:vc} shows that both constraints of Lemma~4.3 of~\cite{TRAP22}---the partition-overlap bound $n>3k/2{+}3t$ and the delivery bound $n>2(k{+}t)$---were conservative: the fork mechanism delivers baiters' commitments to all honest players, and PoB validation restricted to residual RB2 lists prevents fabrication. Removing both constraints retroactively extends the $3f{+}1$ \Trap{} from $C<n/2$ to $C<5n/9$, with the winner consensus feasibility conditions ($C\leq\lfloor(5f{+}2)/3\rfloor$ and $t\leq f$) as the new binding limits. In $5f{+}1$, the no-fork zone and large quorum intersection create a richer regime structure.

\section{Discussion}\label{sec:discussion}

\paragraph{Zero deposits in the no-fork regime.}
In this regime, rational players cannot profit from a final disagreement. Disclosing PoFs earns a positive reward funded by protocol inflation, transaction fees, or any source unrelated to player deposits. The protocol does not need to hold player capital hostage. The $4f{+}1$ quorum structure provides enough safety for liveness-based incentives to dominate without economic penalties.

\paragraph{Byzantine tolerance across regimes.}
The system fixes $n=5f{+}1$ at deployment. The actual values of $k$ and $t$ are unknown: they are whatever materializes at runtime. The protocol guarantees safety as long as the bounds are met, and the regime that applies depends on the actual $(k,t)$. The only deployment-time choice is whether to use simple finalization (targeting the no-fork regime) or \BFTCR{} with deposits (targeting the fork regimes). In the no-fork regime, safety holds for any $C\leq 3f$ with $t\leq f$ ($20\%$), including $k\leq 2f$ ($40\%$). In the fork regimes, safety extends to $C$ up to ${\approx}74\%$ with the \emph{same} Byzantine tolerance $t\leq f$ ($20\%$): the constraint $t<2f/3$ previously stated for the winner consensus is redundant, because the residual Byzantine count is bounded by $C'<2f/3$ regardless of $t$ (Lemma~\ref{lem:wc}). The fork regime therefore adds ${\approx}14\%$ coalition tolerance (from $60\%$ to $74\%$) at no cost in Byzantine tolerance, only requiring deposits.
With the fixed-reward variant described below, the WC is eliminated and the ceiling rises to $C\leq 4f\approx 79.99\%$. Deposits are more expensive ($d>(a{-}1)/(3f{+}1)$, roughly $7$--$15\times$ the lottery version) because the protocol cannot distinguish few-rationals from many-rationals-all-baiting: the slashed pool must fund the worst case. Safety extends to $C$ up to ${\approx}79.99\%$.

\paragraph{Extending beyond the winner consensus.}
The WC serves two logically distinct functions: (a)~resolving which branch is canonical, and (b)~selecting the reward recipient. Function~(a) requires no consensus: given the PoFs and both branches, every honest player independently computes the same canonical branch via an agreed deterministic merge rule~\cite{RG21,Basilic23}, of which awareness of whether disagreement is still possible depends on the number of detected faults and assumption on the maximum size of the tolerated coalition. The entire WC bottleneck is function~(b). The WC can be eliminated by awarding a fixed reward $\RR$ to every valid baiter whose reveal is verified. Since each baiter's RB1 commitment was delivered via reliable broadcast, any honest player receiving a decryption key can locally verify and forward the reveal exactly once, giving implicit dissemination without a separate RB round or consensus on the candidate set. The deposit condition simplifies to $d>(a{-}1)/(3f{+}1)$, independent of $C$, $k$, and $m$. The binding constraint becomes valid-candidacy alone: $C\leq 4f$ (${\approx}79\%$ for $f=20$). For $a\geq 3$ branches, not all honest players may detect the fork in a single \BFTCR{} round (paralleling the confirmation issue in ZLB~\cite{RG21}); iterated execution of \Snare{} with exclusion of detected equivocators resolves this, as in ZLB, with faster convergence in $5f{+}1$ due to the large quorum intersection. A dedicated treatment of iterated \Snare{} with the fixed-reward variant is future work. Alternatively, Basilic's eventual consensus mode~\cite{Basilic23} on the $2f$ residual tolerates $C'<f$ instead of $C'<2f/3$, reaching the same $C\leq 4f$ ceiling with temporary disagreement on the reward; and a trusted third party or soft synchrony assumption can serve as escrow only in the event of a detected disagreement.

\paragraph{Asynchronous compatibility.}
The \BFTCR{} mechanism and the baiting incentive structure are asynchronous: reliable broadcast terminates without timing assumptions, and the deposit, reward, and PoB validation are purely combinatorial. The partial synchrony assumption in this paper enters only through two components: the accountable consensus layer producing predecisions, and the winner consensus selecting the reward recipient, both of which can be replaced by probabilistically live, asynchronous counterparts. Consequently, \Trap{} and \Snare{} can operate under asynchronous consensus with no modification to the \BFTCR{} phase or the game-theoretic analysis. The same observation applies to the original $3f{+}1$ \Trap{}.

\paragraph{Deposit scaling.}
The generic deposit formula has factor $(\bar{a}{-}1)$ in the numerator and $3f{-}m{+}1$ in the denominator. The $3f$ (rather than $f$) comes from the $3f{+}1$ quorum intersection in the $5f{+}1$ model: the large number of detected fraudsters funds a reward $3\times$ larger than the original \Trap{}, making baiting dominant at correspondingly lower deposits. Triple-spend deposits are $2\times$ the double-spend base (the $(\bar{a}{-}1)$ factor). Deposits stay under $0.5\%$ of the gain even at the highest coalitions.

\section{Conclusion}\label{sec:conclusion}

We presented \Snare{}, the adaptation of \Trap{} to the $5f{+}1$ model, and identified a regime structure with qualitative improvements over $3f{+}1$. Below $60\%$ coalition, one extra all-to-all round prevents any final disagreement with zero deposits, a result applicable to any $5f{+}1$ protocol at one message delay. From $60\%$ to $70\%$, double-spending requires deposits under $0.2\%$. From $70\%$ to $74\%$, triple-spending requires under $0.5\%$. The $3\times$ gap between predecision ($20\%$) and finalization ($60\%$) thresholds, combined with the $3f{+}1$ quorum intersection that funds a $3\times$ larger reward, makes the pentagonal model particularly attractive: accountability triggers early, final disagreement remains expensive, and the deposits needed to incentivize betrayal are minimal.

We also showed that valid-candidacy holds unconditionally in both models, retroactively extending the $3f{+}1$ \Trap{} to $5n/9$. The binding constraint is the winner consensus, which can be eliminated by awarding a fixed reward to each valid baiter upon reveal verification, pushing the ceiling to $C\leq 4f\approx 79\%$ at the cost of higher deposits. For greater fault tolerance, iterated exclusion \`a la ZLB~\cite{RG21} can reduce the effective coalition across successive rounds; a dedicated treatment is future work.

\bibliographystyle{plain}

\end{document}